\documentclass[jets]{usenixjournal}

\usepackage{epsfig,endnotes}
\usepackage{times}
\usepackage{url}
\makeatletter
\g@addto@macro{\UrlBreaks}{\UrlOrds}
\makeatother
\usepackage[colorinlistoftodos]{todonotes}
\usepackage{xifthen}
\usepackage{amssymb}
\usepackage{eucal}
\usepackage{multicol}
\usepackage{minibox}
\usepackage{fixltx2e}
\usepackage{tabularx}
\usepackage{marginnote}
\usepackage{mathtools}
\usepackage{xcolor}
\usepackage{multirow}
\usepackage{relsize}
\usepackage{bm}
\usepackage[hidelinks]{hyperref}
\usepackage{cleveref}

\newcolumntype{L}{>{\raggedright\arraybackslash}X}

\DeclareMathAlphabet{\mathcal}{OMS}{cmsy}{m}{n}

\newtheorem{property}{Property}
\newtheorem{claim}{Claim}
\crefalias{AlgoLine}{line}%
\makeatletter
\let\cref@old@stepcounter\stepcounter
\def\stepcounter#1{%
	\cref@old@stepcounter{#1}%
	\cref@constructprefix{#1}{\cref@result}%
	\@ifundefined{cref@#1@alias}%
	{\def\@tempa{#1}}%
	{\def\@tempa{\csname cref@#1@alias\endcsname}}%
	\protected@edef\cref@currentlabel{%
		[\@tempa][\arabic{#1}][\cref@result]%
		\csname p@#1\endcsname\csname the#1\endcsname}}

\makeatother

\crefformat{chapter}{\S#2#1#3}
\crefmultiformat{chapter}{\S\S#2#1#3}{and~#2#1#3}{, #2#1#3}{, and~#2#1#3}

\crefformat{section}{\S#2#1#3}
\crefmultiformat{section}{\S\S#2#1#3}{and~#2#1#3}{, #2#1#3}{, and~#2#1#3}
\crefname{algocf}{alg.}{algs.}
\Crefname{algocf}{Alg.}{Alg.}
\crefname{equation}{eq.}{eqs.}
\Crefname{equation}{Eq.}{Eqs.}
\crefname{table}{table}{tables}
\Crefname{table}{Table}{Table}
\crefname{definition}{def.}{defs.}
\Crefname{definition}{Def.}{Defs.}

\usepackage[linesnumbered,ruled,nofillcomment]{algorithm2e}
\SetKwBlock{KwFunction}{function}{}
\SetKwBlock{KwUpon}{upon}{}
\SetKwBlock{KwOn}{on}{}
\SetKwBlock{KwAfter}{after}{}
\SetKwBlock{KwInitialize}{initialize}{}
\SetKwBlock{KwRepeat}{repeat}{}
\SetKwBlock{KwPeriodically}{periodically}{}
\SetKwBlock{KwState}{state}{}
\SetKwBlock{KwAtomic}{atomically}{}
\SetKwBlock{KwConcurrently}{concurrently}{}
\SetCommentSty{textup} 
\SetKwComment{KwIttayComment}{(}{)} 
\SetKw{KwAnd}{and} 
\SetKw{KwOr}{or}
\SetKw{KwSetTimer}{setTimer}
\SetKwBlock{KwExpTimer}{on timer expiration}{}
\SetKwFor{ForAll}{forall}{}
\makeatletter
\let\oldnl\nl
\newcommand{\nonl}{\renewcommand{\nl}{\let\nl\oldnl}}
\makeatother

\newcommand{\eg}{e.g.\xspace}
\newcommand{\Eg}{E.g.\xspace}
\newcommand{\ie}{i.e.\xspace}

\newcommand{\tStart}[1]{\ensuremath{t^{\textit{prop}}_{#1}}\xspace}
\newcommand{\tEnd}[1]{\ensuremath{t^{\textit{end}}_{#1}}\xspace}
\newcommand{\tExecute}[1]{\ensuremath{\Delta^{\textit{exec}}_{#1}}\xspace}

\newcommand{\dissTime}{\ensuremath{\delta}\xspace}
\newcommand{\commentx}[1]{\textcolor{blue}{/* #1 */}}
\newcommand{\curr}{\ensuremath{\textit{curr}}\xspace} 
\newcommand{\attemptedQC}{\ensuremath{\textit{attemptedQC}}\xspace}
\newcommand{\attemptedTC}{\ensuremath{\textit{attemptedTC}}\xspace} 

\newcommand{\outcast}{\ensuremath{\textbf{multicast}}\xspace}
\newcommand{\multicast}{\outcast}
\newcommand{\send}{\ensuremath{\textbf{send}}\xspace} 
 
\newcommand{\leader}[1]{\ensuremath{\text{Leader}(#1)}\xspace}
\newcommand{\roundDoubling}[1]{\ensuremath{\textit{VD}(#1)}\xspace}
\newcommand{\wishToAdvance}{\ensuremath{\textsf{\smaller[1]wish\_to\_advance}()}\xspace}
\newcommand{\proposeRound}[1]{\ensuremath{{\textsf{\smaller propose\_view}}(#1)}\xspace} 
\newcommand{\proposeView}{\proposeRound}
\newcommand{\newRound}[1]{\ensuremath{\text{``}\textsf{\smaller[1]WISH},#1\text{''}}\xspace} 
\newcommand{\wish}{\newRound}
\newcommand{\vote}[1]{\ensuremath{\text{``} \textsf{\smaller[1]VOTE},#1 \text{''}}\xspace} 
\newcommand{\TC}[1]{%
	\ifthenelse{\isempty{#1}}%
	{\ensuremath{\text{``}\textsf{\smaller[1]TC}\text{''}}\xspace}
	{\ensuremath{\text{``} \textsf{\smaller[1]TC},#1\text{''}}\xspace}
}
\newcommand{\QC}[1]{%
	\ifthenelse{\isempty{#1}}%
	{\ensuremath{\text{``} \textsf{\smaller[1]QC}\text{''}}\xspace}
	{\ensuremath{\text{``} \textsf{\smaller[1]QC},#1\text{''}}\xspace}
}

\newcommand{\honestNodes}{\ensuremath{H}\xspace}

\newcommand{\timeInterval}{\ensuremath{\mathcal{I}}\xspace}
\newcommand{\GST}{\ensuremath{\text{GST}}\xspace}
\newcommand{\initialRound}[1]{\ensuremath{\textit{init}_{#1}}\xspace}
\newcommand{\process}[1]{\ensuremath{\mathcal{P}_{#1}}\xspace}
\newcommand{\rMax}{\ensuremath{v^{\GST}_{\textit{max}}}\xspace}
\newcommand{\sys}{\ensuremath{\text{Cogsworth}}\xspace}

\newcommand{\wishToAdvanceInterval}{\ensuremath{\alpha}\xspace}
\newcommand{\intervalStart}[1]{\ensuremath{#1^{\textit{start}}}\xspace}

\makeatletter

\newcommand{\scriptveryshortarrow}[1][3pt]{{%
		\hbox{\rule[\scriptratio\dimexpr\fontdimen22\textfont2-.2pt\relax]
			{\scriptratio\dimexpr#1\relax}{\scriptratio\dimexpr.4pt\relax}}%
		\mkern-4mu\hbox{\let\f@size\sf@size\usefont{U}{lasy}{m}{n}\symbol{41}}}}

\makeatother

\begin{document}

\title{{\sys: Byzantine View Synchronization}}
\author{Oded Naor
\affil{Technion and Calibra}
Mathieu Baudet
\affil{Calibra}
Dahlia Malkhi
\affil{Calibra}
Alexander Spiegelman
\affil{VMware Research}}

\begin{abstract}
Most methods for Byzantine fault tolerance (BFT) in the partial synchrony setting divide the local state of the nodes into views, 
and the transition from one view to the next dictates a leader change.
In order to provide liveness, all honest nodes need to stay in the same
view for a sufficiently long time. This requires \emph{view synchronization}, a
requisite of BFT that we extract and formally define here.

Existing approaches for Byzantine view synchronization incur quadratic communication (in
$n$, the number of parties). A cascade of $O(n)$ view changes may thus result in $O(n^3)$ communication complexity.
This paper presents a new Byzantine view synchronization algorithm named \sys, 
that has optimistically linear communication complexity and constant
latency. Faced with benign failures, \sys has expected linear communication and
constant latency. 

The result here serves as an important step towards reaching
solutions that have overall quadratic communication, the known lower bound on Byzantine fault tolerant consensus.
\sys is particularly useful for a family of BFT protocols that already exhibit
linear communication under various circumstances, but suffer quadratic overhead
due to view synchronization.
\end{abstract}

\maketitle

\sloppy
\section{Introduction}

Logical synchronization is a requisite for progress to be made in asynchronous state machine replication (SMR). 
Previous Byzantine fault tolerant (BFT) synchronization mechanisms incur quadratic message complexities, frequently dominating
over the linear cost of the consensus cores of BFT solutions.
In this work, we define the \emph{view synchronization} problem and
provide the first solution in the Byzantine setting, whose latency is bounded and communication cost is
linear, under a broad set of scenarios.

\subsection{Background and Motivation}
\label{sec:intro:motivation}

Many practical reliable distributed systems do not rely on network synchrony
because networks go through outages and periods of Distributed Denial-of-Service~(DDoS) attacks; and because
synchronous protocols have hard~coded steps that wait for a maximum delay. 
Instead, asynchronous replication solutions via state machine replication (SMR)~\cite{DBLP:journals/csur/Schneider90} 
usually optimize for stability periods. This approach is
modeled as partial synchrony~\cite{dwork1988consensus}. It allows for periods of asynchrony in which
progress might be compromised, but consistency never does.

In the crash-failure model, 
this paradigm underlies most successful industrial solutions, for example,
the Google Chubbie lock service~\cite{DBLP:conf/osdi/Burrows06}, Yahoo's
tcdZooKeeper~\cite{DBLP:conf/usenix/HuntKJR10}, etcd~\cite{etcd}, Google's Spanner~\cite{DBLP:journals/tocs/CorbettDEFFFGGHHHKKLLMMNQRRSSTWW13}, Apache Cassandra~\cite{cassandra2014apache} and others.
The algorithmic cores of these systems, \eg, 
Paxos~\cite{lamport2001paxos}, Viewstamp Replication~\cite{oki1988viewstamped},
or Raft~\cite{ongaro2014raft}, revolve around a view-based paradigm. 
In the Byzantine model, this paradigm underlies many blockchain systems,
including VMware's Concord~\cite{vmware2019concord}, Hyperledger Fabric~\cite{androulaki2018hyperledger},
Cypherium~\cite{guo2017cypherium,Cypherium2019HotstuffBlog}, Celo~\cite{celo}, PaLa~\cite{chan2018pala} and Libra~\cite{libra2019whitepaper}.  The algorithmic cores of these BFT system are 
view-based, e.g., PBFT~\cite{castro1999practical}, SBFT~\cite{gueta2019sbft}, and HotStuff~\cite{yin2019hotstuff}. 

The advantage of the view-based paradigm is that each view has a designated
leader that can drive a decision efficiently. Indeed, in both models, there are
protocols that have per-view linear message and communication complexity, which is optimal.

In order to guarantee progress, nodes must give up when a view does
not reach a decision after a certain timeout period. Mechanisms for changing the view
whose communication is linear exist both for the crash model (all the above) and,
recently, for the Byzantine model (HotStuff~\cite{yin2019hotstuff}). 
An additional requirement for progress is that all nodes overlap for a
sufficiently long period. 
Unfortunately, all of the above protocols incur quadratic message complexity for view synchronization. 

In order to address this, we first define the \emph{view synchronization}
problem independently of any specific protocol and in a fault-model agnostic manner.
We then introduce a view synchronization algorithm called \sys whose
message complexity is linear in expectation, as well as in the worst-case under
a broad set of conditions.

\subsection{The View Synchronization Problem}
\label{sec:intro:viewSyncronization}

We introduce the problem of~\emph{view synchronization}.  All nodes start at view zero. 
A view change occurs as an interplay between the synchronizer, which implements a view
synchronization algorithm and the outer consensus solutions. The consensus
solution must signal that it wishes to end the current view via a
\emph{\wishToAdvance} notification. The synchronizer eventually invokes a consensus
\emph{\proposeView{v}} signal to indicate when a new view~$v$ starts. 
View synchronization requires to
eventually bring all honest nodes to execute the same view for a sufficiently long time, for the outer consensus protocol to be able to drive progress. 

The two measures of interest to us are latency and
communication complexity between these two events. Latency is measured only
during periods of synchrony, when a bound $\dissTime$ on message transmission delays
is known to all nodes, and is expressed in $\dissTime$ units.

View synchronization extends the PaceMaker abstraction presented in~\cite{yin2019hotstuff}, formally defines the problem it solves, and captures it as a separate component.
It is also related to the seminal work of Chandra and Toueg~\cite{chandra1996weakest}, \cite{chandra1996unreliable} 
on \emph{failure detectors}.
Like failure detectors, it is an abstraction capturing the conditions under which progress is guaranteed,
without involving explicit engineering solutions details such as packet transmission delays, timers, and computation.
Specifically, Chandra and Toueg define a leader election abstraction, denoted~$\Omega$, where eventually all non-faulty nodes trust the same non-faulty node as the leader.
$\Omega$ was shown to be the weakest failure detector needed in order to solve consensus.
Whereas Chandra and Toueg's seminal work focuses on the possibility/impossibility of an eventually elected leader, here we
care about how quickly it takes for a good leader to emerge (\ie, the latency), at what
communication cost, and how to do so repeatedly, allowing the extension of one time single-shot consensus to SMR. 

We tackle the view synchronization problem against asynchrony and
the most severe type of faults, 
Byzantine~\cite{lamport1982byzantine,lamport1983weak}. This makes the
synchronizers we develop particularly suited for Byzantine Fault
Tolerance (BFT) consensus systems relevant in today's crypto-economic systems.

More specifically, we assume a system of $n$ nodes that need to form a sequence of
\emph{consensus} decisions that implement SMR. We assume up to
${f < n/3}$ nodes are Byzantine, the upper bound on the number of Byzantine nodes in which Byzantine agreement is
solvable~\cite{fischer1986easy}.  The challenge is that during ``happy'' periods, 
progress might be made among a group of Byzantine nodes cooperating with a
``fast'' sub-group of the honest nodes. Indeed, many solutions advance when a
leader succeeds in proposing a value to a quorum of $2f+1$ nodes, but it is
possible that only the $f+1$ ``fast'' honest nodes learn it and progress to the
next view.  The remaining $f$ ``slow'' honest nodes
might stay behind, and may not even advance views at all. 
Then at some point, the $f$ Byzantine nodes may stop cooperating.
A mechanism is needed to bring the ``slow'' nodes to the same view as the $f+1$
``fast'' ones.

Thus, our formalism and algorithms may be valuable for the consensus protocols
mentioned above, as well as others, such as Casper~\cite{buterin2017casper} and
Tendermint~\cite{buchman2017tendermint,buchman2018tendermint}, which reported 
problems around
liveness~\cite{milosevic2018TendermintLivenessIssue,Pyrofex2018CasperLivenessIssue}.

\subsection{View Synchronization Algorithms}
\label{sec:intro:algsDescription}

We first extract two synchronization mechanisms that borrow from previous BFT
consensus protocols, casting them into our formalism and analyzing them. 

One is a straw-man mechanism that requires no
communication at all and achieves synchronization albeit with unbounded latency.
This synchronizer works simply by doubling the duration of each view.
Eventually, it guarantees a sufficiently long period in which all the nodes are in the same view.

The second is the broadcast-based synchronization mechanism built into PBFT~\cite{castro1999practical} and similar Byzantine protocols, such as~\cite{gueta2019sbft}.
This synchronizer borrows from the Bracha reliable broadcast algorithm~\cite{bracha1987asynchronous}.
Once a node hears of~$f+1$ nodes who wish to enter the same view, it relays the wish reliably so
that all the honest nodes enter the view within a bounded time.

The properties of these synchronizers in terms of latency and communication costs are summarized in~\Cref{tab:syncComparison}.
For brevity, these algorithms and their analysis are deferred to~\Cref{app:algorithms}.

\subsection*{\sys: leader-based synchronizer}

The main contribution of our work is \sys\footnote{\sys is the enchanted clock from Disney's classic ``Beauty and the Beast''.}, which is a leader-based view synchronization algorithm.
\sys utilizes views that have an honest leader to relay messages, instead of broadcasting them.
When a node wishes to advance a view, it sends the message to the leader of the view, and not to all the other nodes.
If the leader is honest, it will gather the messages from the nodes and
multicast them using a threshold signature~\cite{boneh2001short,cachin2005random,shoup2000practical} to the rest of the nodes, incurring only
a linear communication cost.
The protocol implements additional mechanisms to advance views despite faulty leaders.

The latency and communication complexity of this algorithm depend on the number of actual failures and their type.
In the best case, the latency is constant and communication is linear.
Faced with $t$ benign failures, in expectation, the communication is linear and
in worst case $O(t {\cdot} n)$, as mandated by the lower bound of Dolev and
Reischuk~\cite{dolev1985bounds};
the latency is expected constant and $O(t {\cdot} \dissTime)$ in the worst-case.
Byzantine failures do not change the latency, but they can drive the communication to an expected $O(n^2)$ complexity and in the worst-case up to $O(t{\cdot}n^2)$.
It remains open whether a worst-case linear synchronizer whose latency is constant is
possible.

To summarize, \sys performs just as well as a broadcast-based synchronizer in terms of latency and message complexity, and in certain scenarios shows up-to $O(n)$ better results in terms of message complexity.
\Cref{tab:syncComparison} summarizes the properties of all three synchronizers.
\begin{table}%
	\centering
	\footnotesize
	\tbl{Comparison of the different protocols for view synchronization\label{tab:syncComparison}}{%
	\begin{tabular}{| c | c c l c l l |}
		\hline

		                                  & \multirow{3}{*}{\textbf{View doubling}} & \multicolumn{2}{c}{\multirow{3}{*}{\textbf{broadcast-based}}}    & \multicolumn{3}{c|}{\textbf{\sys:}}                                       \\
		                                  &                                          &           &                                              & \multicolumn{3}{c|}{\textbf{leader-based}}                                       \\
		                                  &                                          &            &                                             & Fault type                                 &            &                        \\
		\hline
		\multirow{5}{*}{\shortstack{\textbf{Communication} \\ \textbf{complexity}}} & \multirow{5}{*}{0} & \multirow{5}{*}{\shortstack{ expected \\ worst-case}} & \multirow{5}{*}{\shortstack{$O(n^2)$ \qquad \\  $O(t{\cdot} n^2)$}} & \multirow{3}{*}{Byzantine} & optimal & $O(n)$ \\
		                                  &                                          &                                                         &                                 &           & expected   & $O(n^2)$               \\ 
		                                  &&&&& worst-case & $O(t{\cdot}n^2)$
		                                  \\ \cline{5-7} 
		                                  &                                          &            &                                             & \multirow{2}{*}{Benign}                    & expected   & $O(n)$                 \\
		                                  &                                          &                                                         &                             &               & worst-case & $O(t {\cdot} n)$               \\
		\hline
		\multirow{2}{*}{\textbf{Latency}} & \multirow{2}{*}{unbounded}               & expected & $O(\dissTime)$                         & \multirow{2}{*}{\shortstack{Byzantine                                            \\ +  Benign}} & expected & $O(\dissTime)$ \\
		                                  &                                          &      worst-cast & $O(t{\cdot}\dissTime)$                                                   &                                            & worst-case & $O(t {\cdot} \dissTime)$ \\
		\hline
	\end{tabular}}
\Note{Note:}{$t$ is the number of actual failures, and \dissTime is the bound on message delivery after \GST.}
\end{table}%
 
\subsection{Contributions}
The contributions of this paper as follows: 

\begin{itemize}
	\item To the best of our knowledge, this is the first paper to formally
define the problem of view synchronization.
	
	\item It includes two natural synchronizers algorithms cast into this framework and uses them as a basis for comparison.

	\item It introduces \sys, a leader-based Byzantine synchronizer exhibiting
	      faultless and expected linear communication complexity and constant latency.
\end{itemize}

\paragraph*{Structure}
The rest of this paper is structured as follows: \Cref{sec:model} discusses the model; \Cref{sec:problemDef} formally presents the view synchronization problem; \Cref{sec:cogsworth} presents the \sys view synchronization algorithm with formal correctness proof latency and communication cost analysis; \Cref{sec:usages} describes real-world implementations where the view synchronization algorithms can be integrated; \Cref{sec:relatedWork} presents related work; and~\Cref{sec:conclusion} concludes the paper.
The description of the two natural view synchronization algorithms, view doubling and broadcast-based, are presented in~\Cref{app:algorithms}.

\section{Model}
\label{sec:model}

We follow the eventual synchronous model~\cite{dwork1988consensus} in which the execution is divided into two durations; first, an unbounded period of asynchrony, where messages do not have a bounded time until delivered; and then, a period of synchrony, where messages are delivered within a bounded time, denoted as~\dissTime.
The switch between the first and second periods occurs at a moment named \emph{Global Stabilization Time~(\GST)}.
We assume all messages sent before GST arrive at or before ${\GST + \dissTime}$ .

Our model consists of a set ${\Pi = \left\lbrace \process{i} \right\rbrace_{i=1}^n}$ of $n$ nodes, and a known mapping, denoted by $\leader{\cdot}$:~${\mathbb{N} \mapsto \Pi}$ that continuously rotates among the nodes. Formally, ${\forall j \geq 0 \colon \bigcup_{i=j}^{\infty} \leader{i} = \Pi}$.
We use a cryptographic signing scheme, a public key infrastructure (PKI) to validate signatures, as well as a threshold signing scheme~\cite{boneh2001short,cachin2005random,shoup2000practical}.
The threshold signing scheme is used in order to create a compact signature of $k$-of-$n$ nodes and is used in other consensus protocols such as~\cite{cachin2005random}.
Usually $k = f+1$ or $k=2f+1$.

We assume a non-adaptive adversary who can corrupt up to~$f < n/3 $ nodes at the beginning of the execution.
This corruption is done without the knowledge of the mapping \leader{\cdot}.
The set of remaining $n-f$ honest nodes is denoted~\honestNodes.
We assume the honest nodes may start their local execution at different times.

In addition, as in~\cite{abraham2019VABA,cachin2005random}, we assume the adversary is polynomial-time bounded, \ie, the probability it will break the cryptographic assumptions in this paper (\eg, the cryptographic signatures, threshold signatures, etc.) is negligible.
\section{Problem Definition}
\label{sec:problemDef}

We define a \emph{synchronizer}, which solves the view synchronization problem, to be a long-lived task
with an API that includes a \emph{\wishToAdvance} operation and
a \emph{\proposeView{v}} signal, where $v \in \mathbb{N}$.
Nodes may repeatedly invoke \emph{\wishToAdvance}, and in return
get a possibly infinite sequence of \emph{\proposeView{\cdot}} signals.
Informally, the synchronizer should be used by a high-level
abstraction (e.g., BFT state-machine replication protocol) to synchronize view numbers in the following
way: All nodes start in view~$0$, and whenever they wish to move to the
next view they invoke \wishToAdvance.
However, they move to view $v$ only when they get a
\emph{\proposeRound{v}} signal.

Formally, a \emph{time interval}~\timeInterval consists of a
starting time~$t_1$ and an ending time~$t_2 \ge t_1$ and all the time
points between them. \timeInterval's length is~$\left| \timeInterval
\right|=t_2-t_1$.
We say $\timeInterval' \subseteq \timeInterval''$ if
$\timeInterval'$ begins after or when $\timeInterval''$ begins, and
ends before or when $\timeInterval''$ ends.
We denote by \tStart{\process{}, v} the time when
node $\process{}$ gets the signal \proposeView{v}, and assume that
all nodes get \proposeView{0} at the beginning of their execution.
We denote time~$t=0$ as the time when the last honest node began its execution, formally $\max_{\process{} \in \honestNodes} \tStart{\process{},0}=0$.
We further denote \tExecute{\process{},v} as the time
interval in which node~$\process{}$ is in view~$v$, \ie, \tExecute{\process{},v}
begins at \tStart{\process{},v} and ends at $\tEnd{\process{},v}
\triangleq \min_{v' > v}\left\lbrace \tStart{\process{},v'}
\right\rbrace$.
We say node~$\process{}$ is at view~$v$ at time~$t$, or
\emph{executes view~$v$ at time~$t$}, if $t \in
\tExecute{\process{},v}$.

We are now ready to define the two properties that any synchronizer must achieve.
The first property, named \emph{view synchronization} ensures that there is an infinite number of views with an honest leader that all the correct nodes execute for a sufficiently long time:

\newcommand{\alphaSynchronizer}{\wishToAdvanceInterval-view synchronizer\xspace}

\begin{property} [View synchronization] \label{prop:roundSync} For
every $c \ge 0$ there exists $\wishToAdvanceInterval  > 0$ and an infinite number
of time intervals and views $\left\lbrace \timeInterval_k, v_k
\right\rbrace_{k=1}^{\infty}$, such that if the interval between every two consecutive calls to
\wishToAdvance by an honest node is 
$\wishToAdvanceInterval$, then for any $k \ge 1$ and any $\process{} \in \honestNodes$ the following holds:
\begin{enumerate}
\item $\left| \timeInterval_k \right| \ge c$
\item $\timeInterval_k \subseteq \tExecute{\process{},v_k}$
\item $\leader{v_k} \in \honestNodes $
\end{enumerate}
\end{property}

The second property ensures that a synchronizer will only signal a new view if an honest node wished to advance to it. Formally:

\begin{property} [Synchronization validity] \label{prop:syncValidity}
	The synchronizer signals
	\proposeView{v'} only if there exists an honest node~$\process{}
	\in \honestNodes$ and some view $v$ s.t.\ $\process{}$ calls
	\wishToAdvance at least $v' - v$ times while executing view $v$.
\end{property}

\paragraph*{Discussion}
The parameter \wishToAdvanceInterval, which is used in~\Cref{prop:roundSync} is the time an honest node waits between two successive invocations of \wishToAdvance, and may differ between view synchronization algorithms.
This parameter is needed to make sure that \wishToAdvance is called an infinite number of times in an infinite run.
In reality, it is likely that in most view synchronization algorithms \wishToAdvanceInterval is larger than some value~$d$ which is a function of the message delivery bound~\dissTime, and also of~$c$ from~\Cref{prop:roundSync}, \ie, the synchronization algorithm will work for any $\wishToAdvanceInterval \geq d \left( \dissTime, c \right)$.
In this case, a consensus protocol using the synchronizer can execute the same view as long as progress is made, and trigger a new view synchronization in case liveness is lost.
See~\Cref{sec:algs:discussion} for concrete examples.

The requirement that the leader of all the synchronized views is honest is needed to ensure that once a view is synchronized, the leader of that view will drive progress in the upper layer protocol, thus ensuring liveness.
Without this condition, a synchronizer might only synchronize views with faulty leaders.

Synchronization validity (\Cref{prop:syncValidity}) ensures that the synchronizer does not suggest a new view to the upper-layer protocol unless an honest node running that upper-layer protocol wanted to advance to that view.

\paragraph*{Latency and communication complexity}
In order to define how the latency and message communication complexity
are calculated, we first define \intervalStart{\timeInterval_k} to be
the time at which the \emph{$k$-th view synchronization is reached}.
Formally, \intervalStart{\timeInterval_k} $\triangleq
\max_{\process{} \in H} \left\lbrace
\tStart{\process{},v_k}\right\rbrace$, where $v_k$ is defined
according to Property 1. 

With this we can define the latency of a synchronizer implementation: 

\begin{definition} [Synchronizer latency] \label{def:laqtency}
	The latency of a synchronizer is defined as $\lim_{n \to \infty} \left( \left( \intervalStart{\timeInterval_1} - \GST \right) +   \sum_{k=2}^{n} \intervalStart{\timeInterval_{k}} - \intervalStart{\timeInterval_{k-1}} \right) / n$.
\end{definition}

Next, in order to define communication complexity, we first need to introduce a few more notations.
Let~$M_{\process{},v_1 \to v_2}$ be the total number of messages \process{} sent between \tStart{\process{}, v_1} and \tStart{\process{},v_2}.
In addition, denote $M_{\process{},\to v}$ as the total number of messages sent by~\process{} from the beginning of \process{}'s execution and~\tStart{\process{},v}.

With this, we define the communication complexity of a synchronizer implementation:
\begin{definition} [Synchronizer communication complexity] \label{def:communicationCost}
	Denote $v_k$ the $k$-th view in which view synchronization
	occurs~(Property 1).
	The message communication cost of a synchronizer is defined as $\lim_{n \to \infty} \left( \sum_{\process{} \in \honestNodes} M_{\process{},\to v_1} + \sum_{k=2}^{n} \sum_{\process{} \in \honestNodes} M_{\process{},v_{k-1} \to v_{k}} \right) / n$.
\end{definition}

This concludes the formal definition of the view synchronization problem.
Next, we present \sys, a view synchronization algorithm with expected constant latency and linear communication complexity in a variety of scenarios.

\section{\sys: Leader-Based Synchronizer}
\begin{algorithm*}[t] \footnotesize
	\caption{\sys: Leader-based synchronizer for node \process{}}
	\label{alg:relibra}
	\SetAlgoNoEnd
	\DontPrintSemicolon
	\SetInd{0.4em}{0.4em}

	\KwInitialize(:\label{alg:relibra:initialize} )  {
		$\curr \gets 0$ \;
		$\attemptedTC \gets 0$ \;
		$\attemptedQC \gets 0$ \;
	}
	
	\BlankLine
	\BlankLine
			
	\nonl \textbf{Every node}: \;
	\KwOn({\wishToAdvance: } \label{alg:relibra:wishToAdvance}) {
		\send \newRound{\curr + 1} to \leader{\curr + 1}
	}
	\BlankLine
	\KwUpon({receiving a valid \TC{v} from \leader{r} s.t. $v \le r \le v+f+1$:	}  \label{alg:relibra:TC}) 
	{
		\send \TC{v} to \leader{v} \quad \commentx{Ensures that if \leader{v} is honest then latency is linear} \label{alg:relibra:forwardTCtoK}\;
		\send \vote{v} to \leader{r}
	}

	\KwAfter({$2\dissTime$ from last sending \newRound{v} \KwAnd not receiving \TC{v} \KwAnd $\attemptedTC \le \curr + f+1$:} \label{alg:relibra:attemptedTC})
	{
		\send \newRound{v} to \leader{\attemptedTC} \;
		$\attemptedTC \gets \attemptedTC + 1$ \;
	}
	\BlankLine
	\KwUpon({receiving a valid \QC{v} from \leader{r} s.t. $v \le r \le v+f+1$:} \label{alg:relibra:QC})
	{
		$\attemptedTC \gets v$ \;
		$\attemptedQC \gets v$ \;
		$\curr \gets v$ \;
		\proposeRound{v} \; \label{alg:relibra:advanceToRound}
	}
	
	\KwAfter({$2\dissTime$ from last sending \vote{v} \KwAnd not receiving \QC{v} \KwAnd $\attemptedQC \le \curr+f+1$:} \label{alg:relibra:attemptedQC})
	{
		\nonl \commentx{Also need to send \TC{v} so that the next leader can multicast to the rest} \;
		\send \vote{v} \KwAnd \TC{v} to \leader{\attemptedQC} \label{alg:relibra:forwardTC}\;
		$\attemptedQC \gets \attemptedQC + 1$ \;
	}

	\BlankLine
	\BlankLine

	\nonl \textbf{Leader (\leader{r} = \process{}):}\;
	\KwUpon({receiving $f + 1$ $\newRound{v}$ \KwOr \TC{v} s.t. $r-(f+1) \le v \le r$ and not sending \TC{v} before:} \label{alg:relibra:leaderReceiveTC}) {
		\multicast $\TC{v}$ with a threshold signature to all nodes (including self) 
	}

	\KwUpon({receiving $2f + 1$ $\vote{v}$ messages s.t. $r-(f+1) \le v \le r$ and not sending \QC{v} before:} \label{alg:relibra:leaderReceiveQC}) {
		\multicast $\QC{v}$ with a threshold signature to all nodes (including self) \label{alg:relibra:leaderMulticastQC}
	}
\end{algorithm*}
\label{sec:algs:sys}
\label{sec:cogsworth}

Before presenting \sys, it is worth mentioning that we assume that all messages between nodes are signed and verified; for brevity, we
omit the details about the cryptographic signatures. In the algorithm, when a
node collects messages from $x$ senders, it is implied that these messages carry $x$ distinct
signatures.
We also assume that the \leader{\cdot} mapping is based on a permutation of the nodes such that every consecutive $f+1$ views have at least one honest leader, \eg, ${\leader{v} = \left( v \bmod n \right) + 1}$.
The algorithm can be easily altered to a scenario where this is not the case.

\subsection{Overview} \sys is
a new approach to view synchronization that leverages leaders to optimistically
achieve linear communication.
The key idea is that instead of nodes broadcasting synchronization messages all-to-all and
incurring quadratic communication, nodes send messages to the leader of the view
they wish to enter. If the leader is honest, it will relay a single broadcast
containing an aggregate of all the messages it received, thus incurring only linear communication.

If a leader of a view~$v$ is Byzantine, it might not help as a relay. In this
case, the
nodes time out and then try to enlist the leaders of subsequent views, one by one, up to view~$v+f+1$,
to help with relaying.
Since at least one of those leaders is honest, one of them will
successfully relay the aggregate. 

The full protocol is presented in~\Cref{alg:relibra}, and is consisted of several message types.
The first two are sent from a node to a leader. They are used to signal to the leader that the node is ready to advance to the next stage in the protocol.
Those messages are named~\wish{v} and \vote{v} where~$v$ is the view the message refers to.

The other two message types are ones that are sent from leaders to nodes.
The first is called~\TC{v} (short for ``Time Certificate'') and is sent when the leader receives~$f+1$ \wish{v} messages; and the second is called~\QC{v} (short for ``Quarum Certificate'') and is sent when the leader receives~$2f+1$ \vote{v} messages.
In both cases, a leader aggregates the messages it receives 
using threshold signatures such that each broadcast message from the leader
contains only one signature. 

The general flow of the protocol is as follows: When~\wishToAdvance is invoked,
the node sends~\wish{v} to~\leader{v}, where~$v$ is the view succeeding
$\curr$ (\Cref{alg:relibra:wishToAdvance}).
Next, there are two options: (i)~If \leader{v} 
forms a~\TC{v}, it broadcast it to all
nodes~(\Cref{alg:relibra:leaderReceiveTC}).
The nodes then respond with~\vote{v} message to the
leader~(\Cref{alg:relibra:TC}) (ii)~Otherwise, if $2\dissTime$
time elapses after sending~\wish{v} to~\leader{v} without receiving~\TC{v}, a node gives up and sends~\wish{v} to the next leader, \ie,
\leader{v+1}~(\Cref{alg:relibra:attemptedTC}). It then waits again~$2\dissTime$
before forwarding~\wish{v} to~\leader{v+2}, and so on, until~\TC{v} is received.

Whenever~\TC{v} has been received, a node sends \vote{v} (even if it did not
send~\wish{v}) to \leader{v}. Additionally, as above, it enlists leaders one by
one until \QC{v} is obtained. Here, the node sends leaders \TC{v} as well as
\vote{v}. 
When a node finally receives~\QC{v} from a leader, it enters view~$v$ immediately~(\Cref{alg:relibra:QC}).

\paragraph{Correctness}

We will prove that~\sys achieves eventual view synchronization~(\Cref{prop:roundSync}) for any $\wishToAdvanceInterval \ge 4\dissTime$ as well as synchronization validity~(\Cref{prop:syncValidity}).
Thus, the claims and lemmas bellow assume this.

We start by proving that if an honest node entered a new view, and the leader of that view is honest, then all the other honest nodes will also enter that view within a bounded time.

\begin{claim}
	\label{claim:relibra:honestLeaderBound}
	After \GST, if an honest node enters view~$v$ at time~$t$, and the
leader of view~$v$ is honest then all the honest nodes enter view~$v$ by $t+4\dissTime$, \ie, if $\leader{v} \in \honestNodes$ then ${\max_{\process{i} \in \honestNodes} \left\lbrace \tStart{\process{i},v}\right\rbrace - \min_{\process{j} \in \honestNodes} \left\lbrace \tStart{\process{j},v} \right\rbrace \le 4\dissTime}$.
\end{claim}

\begin{proof}
	Let~$\process{i}$ be the first honest node that entered view~$v$ at time~$t$.
	\process{i} entered view~$v$ since it received~\QC{v} from~\leader{r} such that~$v \le r \le v+f+1$~(\Cref{alg:relibra:QC}).

	If~$r=v$ then we are done, since when \leader{v} sent~\QC{v} it also sent it to all the other honest nodes~(\Cref{alg:relibra:leaderMulticastQC}), which will be received by $t + \dissTime$, and all the honest nodes will enter view~$v$.

	Next, if~$r > k$ then the only way for~\leader{v} to send~\QC{v} is if it gathered~$2f+1$ \vote{v} messages~(\Cref{alg:relibra:leaderReceiveQC}), meaning at least~$f+1$ of the \vote{v} messages were sent by honest nodes.
	An honest node will send a~\vote{v} message only after first receiving~\TC{v} from \leader{r'} s.t.~$v \le r' \le v+f+1$~(\Cref{alg:relibra:TC}).

	Since when receiving a~\TC{v} an honest node sends the \TC{v} to \leader{v}(\Cref{alg:relibra:forwardTCtoK}), then \leader{v} will receive~\TC{v} by~$t+\dissTime$, will forward it to all other nodes by~$t+2\dissTime$, who will send~\vote{v} to~\leader{v} by~$t+3\dissTime$ and by~$t+4\dissTime$ all honest nodes will receive~\QC{v} from~\leader{v} and enter view~$v$.
\end{proof}

Next, assuming an honest node entered a new view, we bound the time it takes to at least~$f+1$ honest nodes to enter the same view.
Note that this time we do not assume anything on the leader of the new view, and it might be Byzantine. 

\begin{claim}
	\label{claim:relibra:boundByzLeader}
	After~\GST, when an honest node enters view~$v$ at time~$t$, at least~$f+1$ honest nodes enter view~$v$ by~$t+2\dissTime (f+2)$, \ie, after \GST for every~$v$ there exists a group S of honest nodes s.t. $\left| S \right| \ge f+1$ and $\max_{\process{i} \in S} \left\lbrace \tStart{\process{i},v}\right\rbrace - \min_{\process{j} \in S} \left\lbrace \tStart{\process{j},v} \right\rbrace \le 2\dissTime (f+2)$.
\end{claim}

\begin{proof}
	Let~$\process{i}$ be the first node that entered view~$v$ at time~$t$.
	\process{i} entered $v$ since it received \QC{v} from \leader{r} and $v \le r \le v+f+1$~(\Cref{alg:relibra:QC}).
	If~\leader{r} is honest then we are done, since \leader{r} multicasted \QC{v} to all honest nodes~(\Cref{alg:relibra:leaderMulticastQC}), and within \dissTime all honest nodes will also enter view~$v$ by~$t+\dissTime$.

	Next, if \leader{r} is Byzantine, then it might have sent \QC{v} to a subset of the honest nodes, potentially only to~$\process{i}$.
	In order to form a \QC{v}, \leader{r} had to receive $2f+1$ \vote{v} messages~(\Cref{alg:relibra:leaderReceiveQC}), meaning that at least~$f+1$ honest nodes sent \vote{v} to \leader{r}.
	Denote~$S$ as the group of those~$f+1$ honest nodes.

	Each node in~$S$ sent \vote{v} message since it received \TC{v} from \leader{r'} for ${v \le r' \le v+f+1}$~(\Cref{alg:relibra:TC}).
	Note that different nodes in~$S$ might have received~\TC{v} from a different leader, \ie, $\leader{r'}$ might not be the same leader for each node in~$S$.

	After a node in~$S$ sent \vote{v} it will either receive a~\QC{v} within~$2\dissTime$ and enter view~$v$, or timeout after~$2\dissTime$ and send \vote{v} with \TC{v} to \leader{v+1}~(\Cref{alg:relibra:forwardTC}).
	They will continue to do so when not receiving~\QC{v} for the next~$f+1$ views after~$v$.
	This ensures that at least one honest leader will receive \TC{v} after at most~$t+2\dissTime f + \dissTime$.
	Then, this honest leader will multicast the~\TC{v} it received~(\Cref{alg:relibra:leaderReceiveTC}) and at most by $t+2\dissTime (f+1)$, all the honest nodes will receive~\TC{v}. 
	The honest nodes will then send \vote{v} to the honest leader, which will be able to create \QC{v} and multicast it.
	The~\QC{v} will thus be received by all the honest nodes by~$t+2\dissTime(f+2)$ and we are done.
\end{proof}

Next, we show that during the execution, an honest node will enter some new view.

\begin{claim}
	\label{claim:relibra:liveness}
	After \GST, some honest node~$\process{i}$ enters a new view.
\end{claim}

\begin{proof}
	From~\Cref{claim:relibra:boundByzLeader}, if an honest node enters some view~$v$, the time by which at least another~$f$ other honest nodes also enter~$v$ is bounded.
	Eventually, those honest nodes will timeout and \wishToAdvance will be invoked~(\Cref{alg:relibra:wishToAdvance}), which will cause them to send~\wish{v+1} to \leader{v+1}.

	If \leader{v+1} is honest, then it will send a~\TC{v+1} to all the nodes~(\Cref{alg:relibra:leaderReceiveTC}) which will be followed by the leader sending a~\QC{v+1}~(\Cref{alg:relibra:leaderReceiveQC}), and all honest nodes will enter view~$v+1$.

	If \leader{v+1} is not honest then the protocol dictates that the honest nodes that wished to enter~$v+1$ will continue to forward their~\wish{v+1} message to the next leaders (up to~\leader{v+f+1}, \Cref{alg:relibra:attemptedTC}) until each of them  receives~\TC{v+1}.
	This is guaranteed since at least one of those~$f+1$ leaders is honest.

	The same process is then followed for~\QC{v+1}~(\Cref{alg:relibra:attemptedQC}), and eventually all of those~$v+1$ honest nodes will enter view~$v+1$.
\end{proof}
\begin{lemma}
	\sys achieves eventual view synchronization~(\Cref{prop:roundSync}).
\end{lemma}
\begin{proof}
	From~\Cref{claim:relibra:liveness} an honest node eventually will enter a new view, and by~\Cref{claim:relibra:boundByzLeader} at least~$f+1$ honest nodes will enter the same view within a bounded time.
	By applying~\Cref{claim:relibra:liveness} recursively and again, eventually, a view with an honest leader is reached and by~\Cref{claim:relibra:honestLeaderBound} all honest nodes will enter the view within $4\dissTime$.

	Thus, for any $c \ge 0$, if the \sys protocol is run with~$\wishToAdvanceInterval = 4\dissTime + c$ it is guaranteed that all honest nodes will eventually execute the same view for $\left| \timeInterval \right| = c$.
	
	The above arguments can be applied inductively, \ie, there exists an infinite number of such intervals and views in which view synchronization is reached, also ensuring that the views that synchronized also have an honest leader.
\end{proof}

\begin{lemma}
	\sys achieves synchronization validity~(\Cref{prop:syncValidity}).
\end{lemma}

\begin{proof}
	To enter a new view~$v$ a~\QC{v} is needed, which is consisted of~$2f+1$ \vote{v} messages \ie, at least~$f+1$ are from honest nodes.
	An honest node will send~\vote{v} message only when it receives a~\TC{v} message, that requires~$f+1$ \newRound{v} message, meaning at least one of those messages came from an honest node.

	An honest node will send~\wish{v} when the upper-layer protocol invokes \wishToAdvance while it was in view~$v-1$.
\end{proof}

This concludes the proof that \sys is a synchronizer for any~$\wishToAdvanceInterval \ge 4\dissTime$.
Similar to the broadcast-based synchronizer, it allows upper-layer protocols to determine the time they spend in each view.

\paragraph{Latency and communication}
Let~\rMax be the maximum view an honest node is in at~\GST, and let~$X$ denote the number of consecutive Byzantine leaders after~\rMax.
Assuming that leaders are randomly allocated to views, then~$X$ is a random variable of a geometrical distribution with a mean of~$n / (n-f)$.
This means that in the worst case of $t = f = \left\lfloor n/3 \right\rfloor$, then ${\mathbb{E}(X) = (3f+1)/(2f+1) \approx 3/2}$.

Since when~$f+1$ honest nodes at view~$v$ want to advance to view~$v+1$, and if~\leader{v+1} is honest,  all honest nodes enter view~$v+1$ in constant time~(\Cref{claim:relibra:honestLeaderBound}), the latency for view synchronization, in general, is~$O(X {\cdot} \dissTime)$.
For the same reasoning, this is also the case for any two intervals between view synchronizations~(see~\Cref{def:laqtency}).

In the worst-case of~$X = t$, where~$t$ is the number of actual failures during the run, then latency is linear in the view duration, \ie, $O(t {\cdot} \dissTime)$.
But, in the expected case of a constant number of consecutive Byzantine leaders after~\rMax, the expected latency is~$O(\dissTime)$.

For communication complexity, there is a difference between Byzantine failures and benign ones.
If a Byzantine leader of a view~$r$ obtains~\TC{v} for~$r-(f+1) \le v \le r$, then it can forward the~\TC{v} to all the $f+1$ leaders that follow view~$v$ and those leaders will multicast the message~(\Cref{alg:relibra:leaderReceiveTC}), leading to expected~$O(n^2)$ communication complexity, in the case of at least one Byzantine leader after \rMax.
In the worst-case of a cascade of~$t$ failures after \rMax, the communication complexity is~$O(t{\cdot}n^2)$.

In the case of benign failures, communication complexity is dependent on~$X$, since the first correct leader after~\rMax will get all nodes to enter his view and achieve view synchronization, and the benign leaders before it will only cause delays in terms of latency, but will not increase the overall number of messages sent.
Thus, in general, the communication complexity with benign failures is~$O(X {\cdot} n)$.
In the worst-case of~$X = t$ communication complexity is~$O(t {\cdot} n)$, but in the average case it is linear, \ie,~$O(n)$.
For the same reasoning, this is also the case between any consecutive occurrences of view synchronization~(see~\Cref{def:communicationCost}). 

To sum-up, the \emph{expected} latency for both \emph{benign} and \emph{Byzantine} failures is~$O(\dissTime)$, and \emph{worst-case}~${O(t {\cdot} \dissTime)}$.
Communication complexity for \emph{Byzantine} nodes is \emph{optimistically}~$O(n)$, \emph{expected}~${O(n^2)}$ and \emph{worst-case}~$O(t{\cdot}n^2)$ and for benign failures is \emph{expected}~$O(n)$ and \emph{worst-case}~$O(t {\cdot} n)$.

\paragraph*{Discussion}
\sys achieves expected constant latency and linear communication under a broad set of assumptions.
It is another step in the direction of reaching the quadratic communication lower bound of Byzantine consensus in an asynchronous model~\cite{dolev1985bounds}.

In addition to~\sys we present in~\Cref{app:algorithms} two more view synchronization algorithms.
The first one is view doubling, where nodes simply double their view duration when entering a new view, which guarantees that eventually all nodes will be in the same view for sufficiently long.
The other algorithm is borrowed from consensus protocols such as PBFT~\cite{castro1999practical} and SBFT~\cite{gueta2019sbft}.
In~\Cref{sec:algs:discussion} we present a comprehensive discussion on all three algorithms.

\section{Usages and Implementations of Synchronizers}
\label{sec:usages}

In this section, we describe real-world usages of the view synchronization algorithms.
First, it is worth mentioning that many times the terms ``phase'', ``round'', and ``view'' are mixed in different works.
In this work when ``view'' is mentioned, the meaning is that all the nodes agree on some integer value, mapped to a specific node that acts as the leader.

There are SMR protocols where as long as the leader is driving progress in the protocol it is not changed.
This will correspond to all the nodes staying in the same view, and this view can be divided into many phases. \Eg, in PBFT~\cite{castro1999practical} a single-shot consensus consists of two phases. 
In an SMR protocol based on PBFT a view can consist of many more phases, all with the same leader as long as progress is made, and there is no bound on the view duration.

As mentioned in~\Cref{sec:intro:viewSyncronization}, in HotStuff~\cite{yin2019hotstuff}, the view synchronization logic is encapsulated in a module named a PaceMaker, but does not provide a formal definition of what the PaceMaker does, nor an implementation.
The most developed work which adopted HotStuff as the core of its consensus protocol is LibraBFT~\cite{baudet2019librabftV3}.
In LibraBFT, a module also named a PaceMaker is in charge of advancing views.
In this module whenever a node timeouts of its current view, say view~$v$, it sends a message named ``TimeoutMsg, $v$'', and whenever it receives $2f+1$ of these messages, it advances to view~$v$.
In addition, the node sends an aggregated signature of these messages to the leader of view~$v$, which according to the paper, if the leader of~$v$ is honest, guarantees that all other nodes will enter view~$v$ withing~$2 \dissTime$.
The current implementation of the PaceMaker is linear communication as long as there are honest leaders, but quadratic upon reaching a view with a Byzantine one.
The latency is constant.

Many other works on consensus rely on view synchronization as part of their design.
For example, in~\cite{gupta2019proofOfExecution} a doubling view synchronization technique is used: ``For the view-change process, each replica
will start with a timeout~$\delta$ and double this timeout after each view-change (exponential backoff). 
When communication becomes reliable, exponential backoff guarantees that all replicas will eventually
view-change to the same view at the same time.''
\section{Related Work} \label{sec:relatedWork}
\subsubsection*{View synchronization in consensus protocols}
The idea of doubling round duration to cope with partial synchrony borrows from
the DLS work~\cite{dwork1988consensus}, and has been employed in
PBFT~\cite{castro1999practical} and in various works based on DLS/PBFT~\cite{baudet2019librabftV3,buchman2018tendermint,yin2019hotstuff}. 
In these works, nodes double the length of each view when no progress is made.
The broadcast-based synchronization algorithm is also employed as part of the consensus protocol in works such as PBFT.

HotStuff~\cite{yin2019hotstuff} encapsulates view synchronization in a separate
module named a PaceMaker. Here, we
provide a formal definition, concrete solutions, and performance analysis of
such a module.
HotStuff is the core consensus protocol of various works such as Cypherium~\cite{Cypherium2019HotstuffBlog}, PaLa~\cite{chan2018pala} and LibraBFT~\cite{baudet2019librabftV3}.
Other consensus protocols such as Tendermint~\cite{buchman2018tendermint} and Casper~\cite{buterin2017casper} reported issues related to the liveness of their design~\cite{milosevic2018TendermintLivenessIssue,Pyrofex2018CasperLivenessIssue}.

\paragraph*{Notion of time in distributed systems}
Causal ordering is a notion designed to give partial ordering to events in a distributed system.
The most known protocols to provide such ordering are Lamport Timestamps~\cite{lamport1978time} and vector clocks~\cite{fidge1988timestamps}.
Both works assume a non-crash setting.

Another line of work stemmed from Awerbuch work on synchronizers~\cite{awerbuch1985complexity}.
The synchronizer in Awerbuch's work is designed to allow an algorithm that is designed to run in a synchronous network to run in an asynchronous network without any changes to the synchronous protocol itself. 
This work is orthogonal to the work in this paper.

Recently, Ford published preliminary work on Threshold Logical Clocks (TLC)~\cite{ford2019threshold}.
In a crash-fail asynchronous setting, TLC places a barrier on view advancement, \ie, nodes advance to view~$v+1$ only after a threshold of them reached view~$v$.
A few techniques are also described on how to convert TLCs to work in the presence of Byzantine nodes.
The TLC notion of a view ``barrier'' is orthogonal to view synchronization,
though a 2-phase TLC is very similar to our reliable broadcast synchronizer.

\paragraph*{Failure detectors}
The seminal work of Chandra and
Toueg~\cite{chandra1996weakest}, \cite{chandra1996unreliable} introduces the leader election abstraction, denoted ~$\Omega$, and prove it is the weakest failure detector needed to solve consensus.
By using~$\Omega$, consensus protocols can usually be written in a more natural way.
The view synchronization problem is similar to $\Omega$, but differs in several ways. First,
it lacks any notion of leader and isolates the view synchronization component.
Second, view
synchronization adds recurrence to the problem definition. Third, it has a built-in notion of view-duration: nodes commit to spend a constant tine in a view before moving to the next. 
Last, this paper focuses on latency and communication costs of synchronizer
implementations.

\paragraph*{Latency and message communication for consensus}
Dutta et al.~\cite{dutta2007overhead} look at the number of rounds it takes to reach consensus in the crash-fail model after a time defined as GSR~(Global Stabilization Round) which only correct nodes enter. 
This work provides an upper and a lower bound for reaching consensus in this setting.
Other works such as~\cite{alistarh2008solve,dutta2005fast} further discuss the latency for reaching consensus in the crash-fail model.
These works focus on the latency for reaching consensus after \GST.
Both bounds are tangential to our performance measures, as they analyze round latency.

Dolev et al.~\cite{dolev1985bounds} showed a quadratic lower bound on communication complexity to reach Byzantine broadcast, which can be reduced to consensus.
This lower bound is an intuitive baseline for work like ours, though it remains
open to prove a quadratic lower bound on view synchronization per se.

\paragraph*{Clock synchronization}
The clock synchronization problem~\cite{lamport1985synchronizing} in a distributed system requires that the maximum difference between the local clock of the participating nodes is bounded throughout the execution, which is possible since most works assume a synchronous setting.
The clock synchronization problem is well-defined and well-treaded, and there are many different algorithms to ensure this in different models, \eg,~\cite{cristian1989probabilistic,kopetz1987clock,srikanth1987optimal}.
In practical distributed networks, the most prevalent protocol
is~NTP~\cite{mills1991internet}. Again, clock synchronization is an orthogonal
notion to view synchronization, the latter guaranteeing to and stay in
a view within a bounded window, but does not place any bound on the views of different
nodes at any point in time.

\section{Conclusion} \label{sec:conclusion}
We formally defined the \emph{Byzantine view synchronization} problem, which
bridges classic works on failure detectors aimed to solve one-time consensus,
and SMR which consists of multiple one-time consensus instances.
We presented \sys which is a view synchronization algorithm that displays linear communication cost and constant latency under a broad variety of scenarios.

\bibliographystyle{usenixjournal}
\bibliography{references}
\clearpage
\appendix
\section{Protocols for View Synchronization}
\label{sec:relibraAlg}
\label{sec:algorithms}
\label{app:algorithms}

In this section we place into the view synchronization framework two view synchronization algorithms which are used in various consensus protocols, and prove their correctness, as well as discuss their latency and message complexity.

All protocol messages between nodes are signed and verified; for brevity, we
omit the details about the cryptographic signatures. 

\subsection{View Doubling Synchronizer}
\label{sec:algs:roundDoubling}

\newcommand{\currDuration}{\ensuremath{\textit{view\_duration}}\xspace}
\newcommand{\wishVar}{\ensuremath{\textit{wish}}\xspace}
\newcommand{\firstRoundDuration}{\ensuremath{\beta}\xspace}

\begin{algorithm}[t] \footnotesize
	\caption{View doubling synchronizer for node $\process{}$}
	\label{alg:roundDoubling}
	\SetAlgoNoEnd
	\DontPrintSemicolon
	\SetInd{0.4em}{0.4em}

	\KwInitialize(:) {
		$\wishVar \gets 0$ \;
		$\curr \gets 0$ \label{alg:roundDoubling:currInitialize} \;
		$\currDuration \gets \firstRoundDuration $ \commentx{\firstRoundDuration is a predefined value for the duration of the first view} \label{alg:roundDoubnling:durationInitialize} \;
	}

	\BlankLine

	\KwOn({$\wishToAdvance$: } \label{alg:roundDoubling:timeout}) {
		$\wishVar \gets \wishVar +1$ \;
	}
	
	\KwAfter({\currDuration has passed since last changing its value:} \label{alg:roundDoubling:internalTimeout})
	{
		$\curr \gets \curr + 1$ \;
		$\currDuration \gets 2 \times \currDuration$ \;
		\If{ $\wishVar \ge \curr$ \label{alg:roundDoubling:ifCond}} 
		{
			\proposeRound{\curr} \;
		}
	}

\end{algorithm}

\subsubsection{Overview}
A solution approach inspired by PBFT~\cite{castro1999practical} is to use view
doubling as the view synchronization technique.
In this approach, each view has a timer, and if no progress is made the node tries to move to the next view and doubles the timer time for the next view.
Whenever progress is made, the node resets its timer.
This approach is intertwined with the consensus protocol itself, making it hard to separate, as the messages of the consensus protocol are part of the mechanism used to reset the timer.

We adopt this approach and turn it into an independent synchronizer that requires no messages.
Fist, the nodes need to agree on some predefined constant~$\firstRoundDuration > 0$ which is the duration of the first view.
Next, there exists some global view duration mapping $\roundDoubling{\cdot}: \mathbb{N} \mapsto \mathbb{R}^+$, which maps a view~$v$ to its duration: $\roundDoubling{v} =2^v \firstRoundDuration$.
A node in a certain view must move to the next view once this duration passes,
regardless of the outer protocol actions.

The view doubling protocol is described in~\Cref{alg:roundDoubling}.
A node starts at view~$0$~(\Cref{alg:roundDoubling:currInitialize}) and a view duration of~$\firstRoundDuration > 0$~(\Cref{alg:roundDoubnling:durationInitialize}).
Next, when~\wishToAdvance is called, a counter named~\wishVar is incremented~(\Cref{alg:roundDoubling:timeout}).
This counter guarantees validity by moving to a view $v$ only when the
\wishVar counter reaches $v$.
Every time a view ends~(\Cref{alg:roundDoubling:internalTimeout}), an internal counter~\curr is incremented, and if the \wishVar allows it, the synchronizer outputs \proposeRound{v} with a new view~$v$.

\subsubsection{Correctness}
We show that the view doubling protocol achieves the properties required by a synchronizer.

\begin{lemma}
	\label{lem:roundDoublingSynchronization}
	The view doubling protocol achieves view synchronization~(\Cref{prop:roundSync}).
\end{lemma}

\begin{proof}
	Since this protocol does not require sending messages between nodes, the Byzantine nodes cannot affect the behavior of the honest nodes, and we can treat all nodes as honest.

	Recall that $t=0$ denotes the time by which all the honest nodes started their local execution of~\Cref{alg:roundDoubling}.
	Let~\initialRound{i} be the view at which node~$\process{i}$ is at during $t=0$.	
	W.l.o.g assume ${\initialRound{1} \le \initialRound{2} \le \cdots \le \initialRound{n}}$ at time $t=0$.
	It follows from the definition of~\initialRound{i} and the sum of a geometric series that 
	\begin{equation} \label{eq:roundDoubling:tProp}
	\tStart{\process{i},v} = \firstRoundDuration \left( 2^v -2^{\initialRound{i}} \right).
	\end{equation}

	We begin by showing that for every $i \le j$ the following condition holds: $\tStart{\process{i},v} \ge \tStart{\process{j},v}$ for any view~$v$.
	Let $k = \initialRound{i}$ and $l = \initialRound{j}$.
	From the ordering of the node starting times, for all $k \le l$.
	We get:
	\begin{equation*}
		\tStart{\process{i},v} \ge \tStart{\process{j},v} \Leftrightarrow \firstRoundDuration \left( 2^v-2^k \right) \ge \firstRoundDuration \left( 2^v - 2^l \right) \Leftrightarrow l \ge k.
	\end{equation*}
	Hence, for $i \le j$, since at $t=0$ node $\process{j}$ had a view number larger than
$\process{i}$, then $\process{j}$ will start all future views before $\process{i}$. 

	Next, let $k =\initialRound{1}$ and $l = \initialRound{n}$, \ie, the
minimal view and the maximal view at $t=0$ respectively.
	To prove that the first interval of view synchronization is achieved, it suffices to show that for any constant~$c \ge 0$ there exists a time interval~\timeInterval and a view~$v$ such that $\left| \timeInterval \right| \ge c$ and $\tStart{n,v+1} - \tStart{1,v} \ge | \timeInterval |$.
	Using this, we will show that there exists an infinite number of
	such intervals and views that will conclude the proof.
	This also ensures that there is an infinite number of such views with honest leaders.

	Indeed, first note that as shown above, node~$\process{n}$ will start view~$v$ before any other node in the system.
	The left-hand side of the equation is the time length in which both
node~$\process{n}$ and node~$\process{1}$ execute together view~$v$.
	If the left-hand side is negative, then there does not exist an overlap, and if it is positive then an overlap exists.

	We get
	\begin{equation} \label{eq:roundDoubling} 
	\tStart{n,v+1} - \tStart{1,v} \ge | \timeInterval | \Leftrightarrow \firstRoundDuration \left( 2^{v+1} -2^l \right)  - \firstRoundDuration \left( 2^v -2^k \right) \ge | \timeInterval | \Leftrightarrow \firstRoundDuration \left[ 2^v + \left( 2^k -2^l \right) \right] \ge | \timeInterval |.
	\end{equation}

	For any $c \ge 0$ there exists a minimum view number~$v'$ such that the inequality holds, and since $k$ is the minimum view number at $t = 0$ this solution holds for any other node~$\process{i}$ as well.
	In addition, for any $v \ge v'$ the inequality also holds, meaning there is an infinite number of solutions for it, including an infinite number of views with an honest leader.
		
	If \wishToAdvance is called in intervals with $0< \wishToAdvanceInterval
\le \firstRoundDuration$ then by the time the value of \curr reaches some view value~$v$, \wishVar
will always be bigger than \curr, meaning the condition in
\Cref{alg:roundDoubling:ifCond} will always be true, and the synchronizer will
always propose view~$v$ by the time stated in~\Cref{eq:roundDoubling:tProp}.
\end{proof}

\begin{lemma}
	The view doubling protocol achieves synchronization validity~(\Cref{prop:syncValidity}).
\end{lemma}

\begin{proof}
	The if condition in~\Cref{alg:roundDoubling:ifCond} ensures that the output of the synchronizer will always be a view that a node wished to advance to.
\end{proof}

This concludes the proof that view doubling is a synchronizer for any~$0 < \wishToAdvanceInterval \le \firstRoundDuration$.

\subsubsection{Latency and communication}

Since the protocol sends no messages between the nodes, it is immediate
that the communication complexity is $0$.

As for latency, the minimal~$v^*$ satisfying ~\Cref{eq:roundDoubling} grows
with $c \left(2^{\initialRound{n}} - 2^{\initialRound{1}} \right)$. Since the
initial view-gap $\initialRound{n} - \initialRound{1}$ is unbounded, so is the
view $v^*$ in which synchronization is reached. 
The latency to synchronization is $\tStart{\process{1}, v^*} = 2^{v^*} -
2^{\initialRound{1}}$, also unbounded.

\subsection{Broadcast-Based Synchronizer} \label{sec:algs:bracha}
\begin{algorithm}[t] \footnotesize
	\caption{Broadcast-based synchronizer for node \process{}}
	\label{alg:bracha}
	\SetAlgoNoEnd
	\DontPrintSemicolon
	\SetInd{0.4em}{0.4em}

	\KwInitialize(:) {
		$\curr \gets 0$
	}

	\BlankLine
	\BlankLine

	\KwOn({$\wishToAdvance$: } \label{alg:bracha:timeout} ) {
		\outcast $\newRound{\curr + 1}$ to all nodes (including self)
	}
	\BlankLine
	\KwUpon({receiving $f+1$ \newRound{v} messages and not sending \newRound{v} before:} \label{alg:bracha:sendNewRound}) {
		\outcast \newRound{v} to all nodes (including self)
	}

	\BlankLine

	\KwUpon({receiving $2f{+}1$ \newRound{v} messages} \KwAnd $v > \curr$ \label{alg:bracha:advanceRound}) {
		$\curr \gets v$ \;
		\proposeRound{v} 
	}

\end{algorithm}
\subsubsection{Overview}
Another leaderless approach is based on the Bracha reliable broadcast protocol~\cite{bracha1987asynchronous} and is  presented in~\Cref{alg:bracha}.
In this protocol, when a node wants to advance to the next view~$v$ it multicasts a \newRound{v} message (\emph{multicast} means to send the message to all the nodes including the sender)~(\Cref{alg:bracha:timeout}).
When at least $f+1$ \newRound{v} messages are received by an honest node, it multicasts \newRound{v} as well~(\Cref{alg:bracha:sendNewRound}).
A node advances to view~$v$ upon receiving $2f+1$ \newRound{v} messages~(\Cref{alg:bracha:advanceRound}).

\subsubsection{Correctness} \label{test}
We start by showing that the broadcast-based synchronizer achieves eventual view synchronization~(\Cref{prop:roundSync}) for any $\wishToAdvanceInterval \geq 2\dissTime$. Thus, the claims and lemmas below assume this.

\begin{claim}
	\label{claim:brachaAgreement}
	After GST, whenever an honest node enters view~$v$ at time~$t$, all other honest nodes enter view~$v$ by~$t+2\dissTime$, \ie, $\max_{\process{i} \in \honestNodes} \left\lbrace \tStart{\process{i}, v} \right\rbrace - \min_{\process{j} \in \honestNodes} \left\lbrace \tStart{\process{j},v} \right\rbrace \le 2 \dissTime.$
\end{claim}

\begin{proof}
	Suppose an honest node $\process{i} \in \honestNodes$ enters view~$v$ at time~$\tStart{\process{i}, v} = t$, then it received $2f+1$ \newRound{v} messages, from at least~${f+1}$ honest nodes~(\Cref{alg:bracha:advanceRound}).
	
	Since the only option for an honest node to disseminate \newRound{v} message is by multicasting it, then by~$t + \dissTime$ all nodes will receive at least~$f+1$ \newRound{v} messages.
	Then, any left honest nodes (at most~$f$ nodes) will thus receive enough \newRound{v} to multicast the message on their own~(\Cref{alg:bracha:sendNewRound}) which will be received by~$t + 2 \dissTime$ by all the nodes.
	This ensures that all the honest nodes receive $2f+1$ \newRound{v} messages and enter view $v$ by~$t + 2\dissTime$.
\end{proof}

\begin{claim}
	\label{claim:brachaLiveness}
	After GST, eventually an honest node~$\process{i}$ enters some new view.
\end{claim}
\begin{proof}
	All honest nodes begin their local execution at view~$0$, potentially at different times.
	Based on the protocol eventually at least $f+1$ nodes (some of them might be Byzantine) send \wish{1}.
	This is because \wishToAdvance is called every~\wishToAdvanceInterval.
	Thus, eventually all honest nodes will reach view~$1$, and from~\Cref{claim:brachaAgreement} the difference between their entry is at most~$2\dissTime$ after~\GST.
	
	The above argument can be applied inductively.
	Suppose at time~$t$ node~\process{i} is at view~$v$.
	We again know that by~$t+2\dissTime$ all other honest nodes are also at view~$v$, and once $f+1$ \wish{v+1} are sent all honest nodes will eventually enter view~$v+1$, and we are done.	
\end{proof}

\begin{lemma}
	The broadcast-based protocol achieves view synchronization~(\Cref{prop:roundSync}).
\end{lemma}

\begin{proof}
	From~\Cref{claim:brachaLiveness} an honest node will eventually advance to some new view~$v$ and from~\Cref{claim:brachaAgreement} after~$2\dissTime$ all other honest nodes will join it.
	For any~$c \ge 0$, if the honest nodes call \wishToAdvance every~$\wishToAdvanceInterval = 2\dissTime + c$ then it is guaranteed that all the honest nodes will execute view~$v$ together for at least $\left| \timeInterval \right| = c$ time, since it requires $f+1$ messages to move to view~$v+1$, \ie, at least one message is sent from an honest node.
	
	This argument can be applied inductively, and each view after \GST is synchronized, thus making an infinite number of time intervals and views which all honest leaders execute at the same time.
\end{proof}

\begin{lemma}
	The broadcast-based synchronizer achieves synchronization validity~(\Cref{prop:syncValidity}).
\end{lemma}

\begin{proof}
	In order for an honest node to advance to view~$v$ it has to receive $2f+1$ \newRound{v} messages~(\Cref{alg:bracha:advanceRound}).
	From those, at least $f+1$ originated from honest nodes.
	An honest node can send \newRound{v} on two scenarios: 
	
	(i)~\wishToAdvance was called when the node was at view~$v-1$~(\Cref{alg:bracha:timeout}) and we are done.
	
	(ii)~It received $f+1$ \newRound{v} messages~(\Cref{alg:bracha:sendNewRound}), meaning at least one honest node which already sent the message was at view $v-1$ and called \wishToAdvance{} and again we are done.
\end{proof}

This concludes the proof that the broadcast-based synchronizer is a view synchronizer for any~${\wishToAdvanceInterval \ge 2\dissTime}$.

\subsubsection{Latency and communication}
The broadcast-based algorithm synchronizes every view after \GST within $2\dissTime$.
Since the leaders of each view are allocated by the mapping \leader{\cdot}, in expectation every $\approx 3/2$ nodes have an honest leader (see the communication complexity analysis done for \sys in~\Cref{sec:cogsworth}).
Therefore, for latency, the broadcast-based synchronizer will take an expected constant time to reach view synchronization after~\GST, as we have proved, and also the same between every two consecutive occurrences of view synchronization.
Thus, the latency of this protocol is expected~$O(\dissTime)$. 
In the worst-case of $t$ consecutive failures, the latency is~$O(t{\cdot}\dissTime)$.

For communication costs, the protocol requires that every node sends one \wish{v} message to all the other nodes, and since the latency is expected constant, the overall communication costs are also expected quadratic, \ie,~$O(n^2)$.
In the worst-case of~$t$ consecutive failures, the communication complexity is~$O(t{\cdot}n^2)$.

\subsection{Discussion}
\label{sec:algs:discussion}
The three presented synchronizers in the paper have tradeoffs in their latency and communication costs, which are summarized in~\Cref{tab:syncComparison}.
Hence, a protocol designer may choose a synchronizer based on its needs and constraints.
It may be possible to create combinations of the three protocols and achieve
hybrid characteristics; we leave such variations for future work.

In addition, there are differences in the constraints on the parameter~\wishToAdvanceInterval
in these protocols, which is the time interval between two successive calls to~\wishToAdvance~(see~\Cref{prop:roundSync}).
The view doubling synchronizer prescribes a precise~\wishToAdvanceInterval,
which results in each view duration to be exactly twice as its predecessor. 
In the other two synchronizers there is only a lower bound on~\wishToAdvanceInterval: in the
broadcast-based it is~$2\dissTime$, and in \sys it is~$4\dissTime$. 

This difference is significant.
Suppose an upper-layer protocol utilizing the synchronizer wishes to spend an unbounded amount of time in each view as long as progress is made, and triggers a view-change upon detecting that progress is lost.
While the broadcast-based and \sys algorithms allow this upper-layer behavior, the view doubling technique does not, and thus may influence the decision on which view synchronization algorithm to choose.

\end{document}